\documentclass[journal,12pt,onecolumn]{IEEEtran}
\usepackage{amsmath,amsfonts,amsthm,amssymb}
\usepackage{algorithmic}
\usepackage{algorithm}
\usepackage{seqsplit} 
\usepackage{array}
\usepackage[caption=false,font=normalsize,labelfont=sf,textfont=sf]{subfig}
\usepackage{textcomp}
\usepackage{stfloats}
\usepackage{url}
\usepackage{verbatim}
\usepackage{graphicx,marvosym}
\usepackage{color}
\usepackage{cite}
\newtheorem{theorem}{Theorem}
\newtheorem{lemma}{Lemma}

\begin{document}    

\title{Self-Dual Cyclic Codes with Improved Minimum Distance Estimates via Extending the Chen-Ding Construction}

\author{Bofeng Huang, Jingwei Zhang$^*$
 and Chang-An Zhao
\thanks{B. Huang is with the Department of Mathematics, School of Mathematics, Sun Yat-sen University, Guangzhou 510275, P.R.China. (e-mail: huangbf6@mail2.sysu.edu.cn).}
\thanks{J. Zhang is with the
 Department of Big Data Management and Application, 
 Guangdong University of Finance and Economics
 Guangzhou, P. R. China. (e-mail: jingweizhang@gdufe.edu.cn).}
\thanks{C.-A. Zhao is with the School of Mathematics, Sun Yat-sen University, Guangzhou 510275, P.R.China, and also with the Guangdong provincial Key Laboratory of Information Security Technology, Guangzhou 510006, P.R.China (e-mail: zhaochan3@mail.sysu.edu.cn).}
\thanks{$*$ Corresponding author.}
 }

%\markboth{Journal of \LaTeX\ Class Files,~Vol.~1, No.~2, December~2023}%

%\IEEEpubid{0000--0000~\copyright~2023 IEEE}

\maketitle

\begin{abstract}
Self-dual cyclic codes have garnered significant interest owing to their rich algebraic structures and wide-ranging applicability. Their construction and the establishment of lower bounds on their minimum distances are fundamental problems in coding theory. Chen and Ding laid an important foundation for the construction of self-dual cyclic codes in the case where the multiplicative order of $q$ module $n$, denoted by $\operatorname{ord}_n(q)$, is odd. Building on their work, we extend the investigation to the case of even order $\operatorname{ord}_n(q)$ and demonstrate that the minimum distances of the resulting self-dual cyclic codes satisfy square-root lower bounds. By examining the consecutive zero segments in the defining set of the dual code, we determine the exact parameters of Euclidean self-dual cyclic codes with even $\operatorname{ord}_n(q)$ and Hermitian self-dual cyclic codes with odd $\operatorname{ord}_n(q)$. Furthermore, for Euclidean self-dual cyclic codes with odd $\operatorname{ord}_n(q)$ and Hermitian self-dual cyclic codes with even $\operatorname{ord}_n(q)$, we introduce a refined parameter selection that leads to larger minimum distances with the same code length and dimension. This approach also yields tighter lower bounds for several families of self-dual cyclic codes. This work enriches the theory of self-dual cyclic codes and offers new insights into estimating lower bounds on their minimum distances.
\end{abstract}


\section{Introduction}
\label{1}
\subsection{Linear Codes}
\IEEEPARstart{L}{et} $\mathbb{F}_q$ be a finite field of size $q$, where $q$ is a prime power, and let $n$ be a positive integer that is relatively prime to $q$. The \textit{Hamming weight} of a vector $\textbf{a}=(a_0,\dots,a_{n-1})\in \mathbb{F}_q^n$ is defined by 
\[\text{wt}(\textbf{a})=\#\{i:a_i\ne0\}.\] The \textit{Hamming distance} $d(\textbf{a},\textbf{b})$ between two vectors \textbf{a} and \textbf{b} is defined by 
\[d(\textbf{a},\textbf{b})=\text{wt}(\textbf{a}-\textbf{b}).\]
The \textit{minimum Hamming distance} $d(\mathcal{C})$ of a subset $\mathcal{C}\subseteq \mathbb{F}_q^n$ is defined by \[d(\mathcal{C})=\min\limits_{\textbf{a}\ne\textbf{b}}\{d(\textbf{a},\textbf{b}):\textbf{a}\in \mathcal{C},\textbf{b}\in\mathcal{C}\}.\]
A linear code $\mathcal{C}$ over $\mathbb{F}_q$, denoted as $[n,k,d]_q$, is a $k$-dimensional linear subspace of the vector space $\mathbb{F}_q^n$, where $n$ represents the length of code and $d$ denotes its minimum distance. It is clear that the minimum distance of a nonzero linear code equals its minimum nonzero weight~\cite{Blahut2003Algebraic}.

The \textit{Euclidean inner product} on $\mathbb{F}_q^n$ is defined by \[\langle\textbf{x},\textbf{y}\rangle=\sum\limits_{i=0}^{n-1}x_iy_i,\]
where $\textbf{x}=(x_0,\dots,x_{n-1})$ and $\textbf{y}=(y_0,\dots,y_{n-1})$. The \textit{Euclidean dual} of a linear code $\mathcal{C}\subseteq\mathbb{F}_q^n$ is \[\mathcal{C}^\perp=\{\textbf{c}\in \mathbb{F}_q^n:\langle\textbf{c},\textbf{y}\rangle=0,\forall\textbf{y}\in \mathcal{C}\}.\]
A linear code $\mathcal{C}\subseteq\mathbb{F}_q^n$ is \textit{Euclidean self-orthogonal} if $\mathcal{C}\subseteq \mathcal{C}^\perp$, \textit{Euclidean dual-containing} if $\mathcal{C}^\perp\subseteq\mathcal{C}$, and \textit{Euclidean self-dual} if $\mathcal{C}=\mathcal{C}^\perp$.

The Hermitian inner product on $\mathbb{F}_q^n$ is defined as
\[
\langle \mathbf{x}, \mathbf{y} \rangle_H = \sum_{i=0}^{n-1} x_i y_i^{q_1},
\]
whenever $q$ admits the representation $q = q_1^2$ with $q_1$ a prime power; here $\mathbf{x} = (x_0,\dots,x_{n-1})$ and $\mathbf{y} = (y_0,\dots,y_{n-1})$ lie in $\mathbb{F}_q^n$.  The \textit{Hermitian dual} of a linear code $\mathcal{C}\subseteq\mathbb{F}_q^n$ is \[\mathcal{C}^{\perp_H}=\{\textbf{c}\in \mathbb{F}_q^n:\langle\textbf{c},\textbf{y}\rangle_H=0,\forall\textbf{y}\in \mathcal{C}\}.\]
For a linear code $\mathcal{C}\subseteq\mathbb{F}_q^n$, let\[\mathcal{C}^{q_1}=\{(c_0^{q_1},\dots,c_{n-1}^{q_1}):(c_0,\dots,c_{n-1})\in\mathcal{C}\}.\]
Then the Hermitian dual of $\mathcal{C}$ is
\begin{align}
    \label{(1)}
    \mathcal{C}^{\perp_H}=(\mathcal{C}^\perp)^{q_1}=(\mathcal{C}^{q_1})^\perp.
\end{align}
A linear code $\mathcal{C}\subseteq\mathbb{F}_q^n$ is \textit{Hermitian self-orthogonal} if $\mathcal{C}\subseteq \mathcal{C}^{\perp_H}$, \textit{Hermitian dual-containing} if $\mathcal{C}^{\perp_H}\subseteq\mathcal{C}$, and \textit{Hermitian self-dual} if $\mathcal{C}=\mathcal{C}^{\perp_H}$.

\subsection{Cyclic Codes}
An $[n,k,d]_q$ linear code $\mathcal{C}$ is called cyclic if, for any codeword $(c_0,c_1,\dots,c_{n-1})\in\mathcal{C}$, the cyclic shift $(c_{n-1},c_0,\dots,c_{n-2})\in\mathcal{C}$. By associating any vector $(c_0,\dots,c_{n-1})\in\mathbb{F}_q^n$ with the polynomial \[c_0+c_1x+\dots +c_{n-1}x^{n-1}\in\mathbb{F}_q[x]/\langle x^n-1\rangle,\]
any linear code $\mathcal{C}$ of length $n$ over $\mathbb{F}_q$ can be viewed as a subset of the quotient ring $\mathbb{F}_q[x]/\langle x^n-1\rangle$. A code $\mathcal{C}$ is termed cyclic if and only if its corresponding subset in $\mathbb{F}_q[x]/\langle x^n-1\rangle$ forms an ideal in this ring. Moreover, every ideal in $\mathbb{F}_q[x]/\langle x^n-1\rangle$ is principal, then there exists a unique monic polynomial $g(x)$ such that every cyclic code $\mathcal{C}$ can be generated by the polynomial $g(x)$ and the polynomial $g(x)$ divides $x^n-1$. The polynomial $h(x)=(x^n-1)/g(x)$ is known as the parity-check polynomial. The dual code $\mathcal{C}^\perp$ of a cyclic code $\mathcal{C}$ with generator polynomial $g(x)$ is the cyclic code with the generator polynomial $g^\perp(x)=x^{\deg h(x)}h(x^{-1})/h(0)$, which is the reciprocal polynomial of the check polynomial $h(x)$ of $\mathcal{C}$~\cite{Roman1992Coding}. Bose-Chaudhuri-Hocquenghem (BCH) codes, a subclass of cyclic codes, are widely applied in communications and data storage systems due to their simple construction and well‑developed decoding algorithms~\cite{Bose1960a},~\cite{Hocquenghem1959}.

Let $\mathbb{Z}_n=\{0,1,\dots,n-1\}$ denote the ring of integers modulo $n$. For each $i\in \mathbb{Z}_n$, the \textit{q-cyclotomic coset} $C_i$ containing $i$ is defined by
\[C_i=\{iq^j\bmod  n:0\le j\le l-1\},\]
where $l$ is the smallest positive integer such that $iq^l\equiv i\pmod n$ and $a \bmod n$ denotes the unique $b\in \mathbb{Z}_n$ such that $a\equiv b\pmod n$. The smallest integer in $C_i$ is called the \textit{coset leader} of $C_i$.

Let $\Gamma(q,n)$ denote the set of all coset leaders. Then the set $\{C_i:i\in \Gamma(q,n)\}$ is a partition of $\mathbb{Z}_n$. Let $m=\mathrm{ord}_n(q)$ and $\alpha$ be a primitive element of $\mathbb{F}_{q^m}$. Define $\beta=\alpha^{(q^m-1)/n}$. Then $\beta$ is an $n$-th primitive root of unity. Define 
\[m_{\beta^i}(x)=\prod\limits_{j\in C_i} (x-\beta^j).\]
It is easily seen that $m_{\beta^i}(x)$ is an irreducible polynomial in $\mathbb{F}_q[x]$ and the canonical factorization of $x^n-1$ over $\mathbb{F}_q$ is 
\[x^n-1=\prod\limits_{i\in \Gamma(q,n)}m_{\beta^i}(x).\]
By definition, the generator polynomial $g(x)$ of a cyclic code $\mathcal{C}$ of length $n$ over $\mathbb{F}_q$ is a divisor of $x^n-1$. The \textit{defining set} of the cyclic code $\mathcal{C}$ with respect to $\beta$ is defined by 
\[\textbf{T}_g=\{i\in \mathbb{Z}_n:g(\beta^i)=0\}.\]
Then the defining set of a cyclic code is the disjoint union of some $q$-cyclotomic cosets.

Let $\delta\ge 2$ be an integer and $b$ be an integer. Let $\mathcal{C}_{(q,n,\delta,b,\beta)}$ denote the cyclic code over $\mathbb{F}_q$ with length $n$ and generator polynomial \[g(x)=\text{lcm}\{m_{\beta^b}(x),\dots,m_{\beta^{b+\delta-2}}(x)\},\]
where $\text{lcm}$ denotes the least common multiple of the set of polynomials over $\mathbb{F}_q$.  This code is called a \textit{BCH code} with designed distance $\delta$~\cite{huffman2003fundamentals}. It is well known that the minimum distance of the BCH code $\mathcal{C}_{(q,n,\delta,1,\beta)}$ is at least its designed distance $\delta$~\cite{Blahut2003Algebraic}. 

The study of minimum distances of BCH codes has long been a central topic in coding theory. In 2015, Ding, Du, and Zhou formulated a significant conjecture regarding a class of binary BCH codes, establishing a lower bound and conjecturing that it actually gives the exact minimum distance for a certain parameter regime~\cite{7056532}. Subsequently, in 2017, Li carried out an in-depth investigation of the minimum distances of some narrow-sense primitive BCH codes, achieving substantial progress on this problem while also raising three open problems concerning the parameters of such codes~\cite{2017The}. Nearly a decade later, in 2025, Sun settled both the conjecture of Ding, Du, and Zhou and the open problems of Li by determining the exact parameters of a family of narrow-sense primitive BCH codes, thereby confirming the tightness of the previously established bounds~\cite{Sun2025}. In 2026, Chen, Chen, Ding, and Lao studied the minimum distances of several families of BCH codes and derived new lower bounds that improve on the classical BCH bound in certain parameter regimes~\cite{11505961}. In the same year, Tiwari and Kewat determined the exact minimum distances of three classes of primitive BCH codes and further extended their results to certain families of cyclic codes, providing a more precise characterization of the distance properties for these code classes~\cite{11428251}.

\subsection{Related Work and Contributions of This Paper}
Since the 1980s, the study of self-dual cyclic codes has witnessed major advances in existence criteria, constructive approaches, and parameter determination.

The investigation of binary self-dual cyclic codes began with the seminal work of Sloane and Thompson in 1983, who systematically characterized their algebraic structure and established the fundamental necessary and sufficient conditions for their existence~\cite{1056682}. A Plotkin-sum-based construction was introduced for binary repeated-root cyclic codes, which later became a crucial tool for generating self-dual cyclic codes~\cite{75250}. Subsequently, Kai and Zhu extended the study to fields of even characteristic, proving that self-dual cyclic codes exist over such fields if and only if the characteristic is two~\cite{2008On}. A significant breakthrough in bounding the minimum distance was made by Heijne and Top, who constructed a family of binary self-dual cyclic codes with lower bound $d \ge \frac{1}{2}\sqrt{n+2}$---the first result linking the minimum distance directly to the square root of the code length~\cite{5290289}.  A complete enumeration of such codes was provided by analyzing their generator polynomials~\cite{5730592}. More recently, Zhang conducted a thorough quantitative analysis of self-dual cyclic codes over finite fields~\cite{2024The}, and in 2025, Chen and Ding significantly generalized van Lint's theorem, yielding a new family of codes with square-root-like lower bounds~\cite{10856234}. Particularly, an open problem was resolved in~~\cite{wang2026constructionsselfdualbinarycyclic} that had remained unsolved for 70 years: the existence of infinite families of self-dual binary cyclic codes with minimum-distance lower bounds exceeding the square-root lower bounds.

The requirement of characteristic two for self-dual cyclic codes has motivated researchers to explore broader classes, including self-dual negacyclic and quasi-cyclic codes. Xie, Chen, Ding, and Sun  constructed several families of $q$-ary self-dual negacyclic codes of lengths $n$ with minimum distances $d \ge \sqrt{n}$ for various lengths $n$ and any given odd prime power $q$ \cite{10478023}. Additionally, Kawaguchi and Matsui proposed a search framework for binary self-dual quasi-cyclic codes with large minimum weight \cite{9366205}. Their method employs generator polynomial matrices to represent quasi-cyclic codes and combines modulus factorization with the Chinese remainder theorem to efficiently search for codes with large minimum weight.

The main contributions of this paper are summarized as follows:

\begin{itemize}
    \item We provide a new choice of the designed distance $\delta$ and investigate the consecutive zeros segment in the defining sets of both Euclidean and Hermitian dual codes under the new designed distance. Moreover, we derive tighter lower bounds on their minimum distances.
    \item We refine the parameters of the self-dual cyclic codes previously obtained by Chen and Ding in~\cite{10856234} and prove that for Hermitian self-dual cyclic codes with odd $\mathrm{ord}_n(q)$, the actual minimum distance exceeds a square-root lower bound.
    \item We construct Euclidean self-dual cyclic codes with even $\mathrm{ord}_n(q)$ and establish that they exceed square-root lower bound. Specifically, the Euclidean self-dual cyclic code $\mathcal{C}'$ in Theorem~\ref{even Euclidean self-dual} has parameters
    \[
    [2(q^m - 1),\, q^m - 1,\,2(q^{\frac{m}{2}} - 1)]_q.
    \]
    \item We construct Hermitian self-dual cyclic codes with even $\mathrm{ord}_n(q)$ and show that they have square-root lower bounds. The Hermitian self-dual cyclic code $\mathcal{C}'$ in Theorem~\ref{new Hermitian self-dual} has parameters
    \[
    [2(q^m - 1),\, q^m - 1,\, d \ge q_1^{m+1} - q_1^2 + 2]_q \quad (m\ge 4)
    \] and
    \[
    [2(q^m - 1),\, q^m - 1,\, d \ge q_1^3 - q_1^2 + q_1 + 1]_q \quad (m=2).
    \]
\end{itemize}
The parameters of the codes obtained in this paper and those of existing codes are summarized in Tables~\ref{Parameters of Euclidean Self-Dual Cyclic Codes} and~\ref{Parameters of Hermitian Self-Dual Cyclic Codes}, where the dash ``——'' indicates that the corresponding parameter is not applicable.

\begin{table}[htbp]
    \centering
    \caption{Parameters of Euclidean Self-Dual Cyclic Codes with Square-Root Lower Bounds  }
    \label{Parameters of Euclidean Self-Dual Cyclic Codes}
    \begin{tabular}{c c c c c c}
        \hline
        $q$ & $m$ & $\delta$ & Minimum Distance $d$ & Reference \\
        \hline
        $2^s,s\ge1$ & $\ge3$, odd & $q^\frac{m+1}{2}-q+1$ & $\ge q^\frac{m+1}{2}-q+1$ & \cite{10856234} \\
        $2$ & $\ge 10$, even & —— & $\ge 2^{\frac{m}{2}+1}-2^{t+1}-2$ & \cite{wang2026constructionsselfdualbinarycyclic}\\
        $2$ & $\ge 6$, even & —— & $\ge 2^{\frac{m}{2}+1}-2$ & \cite{wang2026constructionsselfdualbinarycyclic} \\
        $2$ & $\ge 8$, even & —— & $\ge 2^\frac{m}{2}+2^{\frac{m}{2}-1}$ & \cite{wang2026constructionsselfdualbinarycyclic} \\ 
        $2$ & $\ge 11$, odd & —— & $\ge 2^{\frac{m+1}{2}}+2$ & \cite{wang2026constructionsselfdualbinarycyclic} \\        
        $2$ & $\ge5$, odd & $2^\frac{m+1}{2}-4s-5<\delta<2^\frac{m+1}{2}-4s-1$ & $\ge \min\{2\delta,2^\frac{m+1}{2}+s+[s~\text{odd}]\}$ & \cite{wang2026constructionsselfdualbinarycyclic} \\
        $2^s,s\ge 2$ & $\ge 5$, odd & $q^\frac{m+1}{2}-2q+1$ & $\ge q^\frac{m+1}{2}+3$  & Theorem \ref{4 small delta self-dual cyclic} \\
        $2^s, s\ge 3$ & $3$ & $q^2-2q+1$ & $\ge q^2+4$  & Theorem \ref{4 small delta self-dual cyclic} \\
        $4$ & $3$ & $9$ & $\ge 18$ & Theorem \ref{4 small delta self-dual cyclic} \\
        $2$ & $\ge 7$, odd  & $2^\frac{m+1}{2}-3$ & $\ge 2^\frac{m+1}{2}$ & Theorem \ref{2 small delta self-dual cyclic} \\
        $2$ & $5$ & $5$ & $\ge 10$ & Theorem \ref{2 small delta self-dual cyclic} \\
        $2^s,s\ge1$ & $\ge3$, odd & $q^\frac{m+1}{2}-q+1$ & $\ge q^\frac{m+1}{2}-q+2$ & Theorem \ref{new Euclidean odd} \\
        $2^s,s\ge1$ & $\ge 2$, even & $q^\frac{m}{2}-1$ & $2(q^\frac{m}{2}-1)$ & Theorem \ref{exact Euclidean} \\
        \hline
    \end{tabular}
\end{table}

\begin{table}[htbp]
    \centering
    \caption{Parameters of Hermitian Self-Dual Cyclic Codes with Square-Root Lower Bounds }
    \label{Parameters of Hermitian Self-Dual Cyclic Codes}
    \begin{tabular}{c c c c c c}
        \hline
        $q_1$ & $m$ & $\delta$ & Minimum Distance $d$ & Reference \\
        \hline
        $2^s,s\ge1$ & $\ge3$, odd & $q_1^m-1$ & $\ge q_1^m-1$ & \cite{10856234} \\
        $2^s,s\ge2$ & $\ge 4$, even & $q_1^{m+1}-2q_1^2+1$ & $\ge q_1^{m+1}+3$ & Theorem \ref{Hermitian low delta} \\
        $2^s, s\ge 3$ & $2$ & $q_1^3-2q_1^2+1$ & $\ge q_1^3+2q_1+4$ & Theorem \ref{Hermitian low delta} \\
        $4$ & $2$ & $33$ & $\ge 66$ & Theorem \ref{Hermitian low delta} \\
        $2^s,s\ge1$ & $\ge 4$, even & $q_1^{m+1}-q_1^2+1$ & $\ge q_1^{m+1}-q_1^2+2$ & Theorem \ref{new Hermitian self-dual} \\
        $2^s,s\ge1$ & $2$, even & $q_1^3-q_1^2+1$ & $\ge q_1^{3}-q_1^2+q_1+1$ & Theorem \ref{new Hermitian self-dual} \\
        $2^s,s\ge1$ & $\ge3$, odd & $q_1^m-1$ & $2(q_1^m-1)$ & Theorem \ref{Exact Hermitian} \\
        \hline
    \end{tabular}
\end{table}

\subsection{The Organization of This Paper}

The remainder of this paper is organized as follows. Section~\ref{2} introduces the preliminaries and basic principles underlying the construction of self-dual cyclic codes. In Section~\ref{3}, we refine the lower bound estimation of the minimum distance for two classes of self-dual cyclic codes and examine the variation of consecutive zeros in the defining set of the dual code when the designed distance $\delta$ is small, thereby deriving tighter lower bounds on the minimum distance of the self-dual cyclic codes. Section~\ref{4} leverages known results on the defining set of the dual code to determine the exact parameters of two classes of self-dual cyclic codes, and consequently establishes that their minimum distances exceed square-root lower bounds. Section~\ref{5} concludes the paper with a summary of the main results.

\section{Preliminaries}
\label{2}
This section introduces the prerequisites for constructing self-dual cyclic codes. We begin by explaining why self-dual cyclic codes are necessarily repeated-root cyclic codes. Since $\dim(\mathcal{C}^\perp) = n - \dim(\mathcal{C})$, the length $n$ of a self-dual cyclic code $\mathcal{C}$ must be even. The following lemma establishes that the characteristic and the code length are not relatively prime.

\begin{lemma}[\cite{2008On}]
    Self-dual cyclic codes over $\mathbb{F}_q$ exists if and only if $q$ is even.
\end{lemma}

Suppose $\mathcal{C} \subseteq \mathbb{F}_q^n$ is a self-dual cyclic code. Then $\gcd(q,n) \ne 1$, which implies that $\mathcal{C}$ is a repeated-root cyclic code. To construct such codes, we adopt the approach of van Lint, in which the Plotkin sum is employed.

Let $\mathcal{C}_1$ and $\mathcal{C}_2$ be $[n,k_1,d_1]_q$ and $[n,k_2,d_2]_q$ linear codes, respectively.
The Plotkin sum of $\mathcal{C}_1$ and $\mathcal{C}_2$ is denoted by $\text{Plotkin}(\mathcal{C}_1,\mathcal{C}_2)$ and defined by
\[
\text{Plotkin}(\mathcal{C}_1,\mathcal{C}_2)=\{(\textbf{u},\textbf{u}+\textbf{v}):\textbf{u}\in \mathcal{C}_1,\,\textbf{v}\in\mathcal{C}_2\},
\]
which can be regarded as a linear code with parameters $[2n,k_1+k_2,\min\{2d(\mathcal{C}_1),d(\mathcal{C}_2)\}]_q$~\cite{MacWilliamsSloane1977}.

\begin{lemma}[\cite{10856234}]
\label{van lint}
    Let $q$ be a power of $2$ and $n$ be an odd positive integer. Let $\mathcal{C}_1\subseteq\mathbb{F}_q^n$ be a cyclic code with generator polynomial $g_1(x)\in\mathbb{F}_q[x]$ and let $\mathcal{C}_2\subseteq\mathbb{F}_q^n$ be a cyclic code with generator polynomial $g_1(x)g_2(x)/f(x)\in\mathbb{F}_q[x]$, where $g_2(x)$ is a divisor of $x^n+1$ and $f(x)=\gcd(g_1(x),g_2(x))$. Then the code $\mathcal{C}=\text{Plotkin}(\mathcal{C}_1,\mathcal{C}_2)$ is permutation-equivalent to the repeated-root cyclic code $\mathcal{C}'$ of length $2n$ generated by the polynomial $g_1^2(x)g_2(x)/f(x)$.
\end{lemma}

By Lemma~\ref{van lint}, the cyclic code $\mathcal{C}_2$ is a subcode of $\mathcal{C}_1$. Furthermore, the invariance of inner products under coordinate permutations implies that the linear code $\mathcal{C}$ is Euclidean (or Hermitian) self-dual if and only if the repeated-root cyclic code $\mathcal{C}'$ is. The following lemma provides sufficient conditions for the Plotkin sum of two codes to be self-dual.

\begin{lemma}[\cite{10856234}]
\label{self-dual}
    Let $q$ be a power of $2$ and $n$ be odd.
    \begin{itemize}
        \item Suppose that linear code $\mathcal{C}\subseteq \mathbb{F}_q^n$ is Euclidean dual-containing. Then the linear code $\text{Plotkin}(\mathcal{C},\mathcal{C}^\perp)$ is Euclidean self-dual and has minimum distance $\min\{2d(\mathcal{C}),d(\mathcal{C}^\perp)\}$.
        \item Suppose that linear code $\mathcal{C}\subseteq\mathbb{F}_q^n$ is Hermitian dual-containing. Then the linear code $\text{Plotkin}(\mathcal{C},\mathcal{C}^{\perp_H})$ is Hermitian self-dual and has minimum distance $\min\{2d(\mathcal{C}),d(\mathcal{C}^{\perp_H})\}$.
    \end{itemize}
\end{lemma}

Next, we use the defining set of cyclic codes to derive a criterion for a cyclic code to be Euclidean (or Hermitian) dual-containing. For any subset $\textbf{T}\subseteq \mathbb{Z}_n$, define 
\[\textbf{T}^{-1}=\{n-i:i\in \textbf{T}\}\]
and 
\[\textbf{T}^c=\mathbb{Z}_n\setminus \textbf{T}.\]
In the case of $q=q_1^2$, define 
\[\textbf{T}^{-q_1}=\{(n-q_1i)\bmod n:i\in\textbf{T}\}.\]

\begin{lemma}[\cite{4106108}]
\label{defining set}
    Let $\mathcal{C}$ be a cyclic code with generator polynomial $g(x)$. The following hold:
    \begin{itemize}
        \item The cyclic code $\mathcal{C}$ is Euclidean dual-containing if and only if $\textbf{T}_g\cap \textbf{T}_g^{-1}=\emptyset$.
        \item The cyclic code $\mathcal{C}$ is Hermitian dual-containing if and only if $\textbf{T}_g\cap\textbf{T}_g^{-q_1}=\emptyset$.
    \end{itemize}
\end{lemma}

Lemma~\ref{defining set} establishes that the defining set of the Euclidean dual code $\mathcal{C}^{\perp}$ is $(\mathbf{T}_g^{-1})^c$, while that of the Hermitian dual code $\mathcal{C}^{\perp_H}$ is $(\mathbf{T}_g^{-q})^c$.

Combining Lemmas~\ref{van lint}, \ref{self-dual} and \ref{defining set}, it suffices to construct a Euclidean (or Hermitian) dual-containing cyclic code. Then, applying this code and its dual to the Plotkin sum yields a Euclidean (or Hermitian) self-dual code. A natural approach is to construct a primitive narrow-sense BCH code that is Euclidean (or Hermitian) dual-containing, which is shown in the following lemmas.

\begin{lemma}[\cite{4106108}]
    \label{Euclidean dual-containing}
    Let $q$ be a prime power, $n=q^m-1$ and $\beta$ be an $n$-th primitive root of unity in $\mathbb{F}_{q^m}$. The following holds:
    \begin{itemize}
        \item If $m$ is odd and the designed distance $\delta$ is in the range $2\le \delta\le q^{\frac{m+1}{2}}-q+1$, then $\mathcal{C}_{(q,n,\delta,1,\beta)}$ is Euclidean dual-containing.
        \item If $m$ is even and the designed distance $\delta$ is in the range $2\le \delta\le q^{\frac{m}{2}}-1$, then $\mathcal{C}_{(q,n,\delta,1,\beta)}$ is Euclidean dual-containing.
    \end{itemize}
\end{lemma}

\begin{lemma}[\cite{4106108}]
    \label{Hermitian dual-containing}
    Let $q_1$ be a prime power and $q=q_1^2$. Let $n=q^m-1$ and $\beta$ be an $n$-th primitive root of unity in $\mathbb{F}_{q^m}$. The following holds:
    \begin{itemize}
        \item If $m$ is odd and the designed distance $\delta$ is in the range $2\le \delta\le q_1^m-1$, then $\mathcal{C}_{(q,n,\delta,1,\beta)}$ is Hermitian dual-containing.
        \item If $m$ is even and the designed distance $\delta$ is in the range $2\le \delta\le q_1^{m+1}-q_1^2+1$, then $\mathcal{C}_{(q,n,\delta,1,\beta)}$ is Hermitian dual-containing.
    \end{itemize}
    
\end{lemma}

Building upon the results of Chen and Ding~\cite{10856234}, we can construct self-dual cyclic codes with the following parameters.

\begin{lemma}[\cite{10856234}]
\label{odd Euclidean self-dual}
    Let $q$ be a power of $2$ and $m$ be odd. Let $n=q^m-1$ and $\beta$ be an $n$-th primitive root of unity in $\mathbb{F}_{q^m}$. Put $\delta=q^{\frac{m+1}{2}}-q+1\ge2$. Let $g^\perp(x)=x^{\deg h(x)}h(x^{-1})/h(0)$ be the generator polynomial of $\mathcal{C}_{(q,n,\delta,1,\beta)}^\perp$, where $g(x)$ is the generator polynomial of $\mathcal{C}_{(q,n,\delta,1,\beta)}$ and $h(x)=(x^n-1)/g(x)$. Let $\mathcal{C}'$ denote the cyclic code of length $2n$ over $\mathbb{F}_q$ with generator polynomial $g(x)g^\perp(x)$. Then $\mathcal{C}'$ is a Euclidean self-dual cyclic code with parameters $[2n,n,d]_q$, where \[d=\min\{2d(\mathcal{C}_{(q,n,\delta,1,\beta)}),d(\mathcal{C}_{(q,n,\delta,1,\beta)}^\perp)\}\ge q^{\frac{m+1}{2}}-q+1.\]
\end{lemma}

\begin{lemma}[\cite{10856234}]
\label{odd Hermitian self-dual}
    Let $q_1$ be a power of $2$, $q=q_1^2$ and $m$ be odd. Let $n=q^m-1$ and $\beta$ be an $n$-th primitive root of unity in $\mathbb{F}_{q^m}$. Put $\delta=q_1^m-1\ge2$. Let $g^{\perp_H}(x)$ be the generator polynomial of $\mathcal{C}_{(q,n,\delta,1,\beta)}^{\perp_H}$. Let $\mathcal{C}'$ denote the cyclic code of length $2n$ over $\mathbb{F}_q$ with generator polynomial $g(x)g^{\perp_H}(x)$, where $g(x)$ is the generator polynomial of $\mathcal{C}_{(q,n,\delta,1,\beta)}$. Then $\mathcal{C}'$ is a Hermitian self-dual cyclic code with parameters $[2n,n,d]_q$, where \[d=\min\{2d(\mathcal{C}_{(q,n,\delta,1,\beta)}),d(\mathcal{C}_{(q,n,\delta,1,\beta)}^\perp)\}\ge q_1^m-1.\]
\end{lemma}

The following theorems are established by following the approach of Chen and Ding~\cite{10856234}.

\begin{theorem}
\label{even Euclidean self-dual}
    Let $q$ be a power of $2$ and $m$ be even. Let $n=q^m-1$ and $\beta$ be an $n$-th primitive root of unity in $\mathbb{F}_{q^m}$. Put $\delta=q^{\frac{m}{2}}-1\ge2$. Let $g^\perp(x)=x^{\deg h(x)}h(x^{-1})/h(0)$ be the generator polynomial of $\mathcal{C}_{(q,n,\delta,1,\beta)}^\perp$, where $g(x)$ is the generator polynomial of $\mathcal{C}_{(q,n,\delta,1,\beta)}$ and $h(x)=(x^n-1)/g(x)$. Let $\mathcal{C}'$ denote the cyclic code of length $2n$ over $\mathbb{F}_q$ with generator polynomial $g(x)g^\perp(x)$. Then $\mathcal{C}'$ is a Euclidean self-dual cyclic code with parameters $[2n,n,d]_q$, where \[d=\min\{2d(\mathcal{C}_{(q,n,\delta,1,\beta)}),d(\mathcal{C}_{(q,n,\delta,1,\beta)}^\perp)\}\ge q^{\frac{m}{2}}-1.\]
\end{theorem}
\begin{proof}
    Let $\mathcal{C}=\text{Plotkin}(\mathcal{C}_{(q,n,\delta,1,\beta)},\mathcal{C}^\perp_{(q,n,\delta,1,\beta)})$.
    By Lemma \ref{van lint}, $\mathcal{C}$ is permutation-equivalent to $\mathcal{C}'$. The desired conclusions then follow from Lemmas \ref{self-dual} and \ref{Euclidean dual-containing}.
\end{proof}

\begin{theorem}
    \label{even Hermitian self-dual}
    Let $q_1$ be a power of $2$, $q=q_1^2$ and $m$ be even. Let $n=q^m-1$ and $\beta$ be an $n$-th primitive root of unity in $\mathbb{F}_{q^m}$. Put $\delta=q_1^{m+1}-q_1^2+1\ge 2$. Let $g^{\perp_H}(x)$ be the generator polynomial of $\mathcal{C}_{(q,n,\delta,1,\beta)}^{\perp_H}$. Let $\mathcal{C}'$ denote the cyclic code of length $2n$ over $\mathbb{F}_q$ with generator polynomial $g(x)g^{\perp_H}(x)$, where $g(x)$ is the generator polynomial of $\mathcal{C}_{(q,n,\delta,1,\beta)}$. Then $\mathcal{C}'$ is a Hermitian self-dual cyclic code with parameters $[2n,n,d]_q$, where \[d=\min\{2d(\mathcal{C}_{(q,n,\delta,1,\beta)}),d(\mathcal{C}_{(q,n,\delta,1,\beta)}^\perp)\}\ge q_1^{m+1}-q_1^2+1.\]
\end{theorem}
\begin{proof}
    Notice that $d(\mathcal{C}^{\perp_H}_{(q,n,\delta,1,\beta)})=d(\mathcal{C}^\perp_{(q,n,\delta,1,\beta)})$. Let $\mathcal{C}=\text{Plotkin}(\mathcal{C}_{(q,n,\delta,1,\beta)},\mathcal{C}^{\perp_H}_{(q,n,\delta,1,\beta)})$.
    By Lemma \ref{van lint}, $\mathcal{C}$ is permutation-equivalent to $\mathcal{C}'$. The desired conclusions then follow from Lemmas \ref{self-dual} and \ref{Hermitian dual-containing}.
\end{proof}

To determine the lengths of consecutive zeros in the defining sets of the Euclidean (or Hermitian) dual codes, we shall need the following lemmas.

\begin{lemma}[\cite{gong2022dual}]
    \label{Euclidean dual BCH}
    For $3 \leq \delta < (q-1)q^{m-1} - q^{\left\lfloor \frac{m-1}{2} \right\rfloor}$, let $I(\delta) \geq 2$ be the integer such that $\{0,1,\dots,I(\delta)-1\} \subseteq \textbf{T}_{g^\perp}$ and $I(\delta) \notin \textbf{T}_{g^\perp}$. Then
    \[
    I(\delta)=
    \begin{cases}
        q^{m-t}-a, & \text{if } aq^t \leq \delta \leq (a+1)q^t-1(1 \leq t \leq m-2,1 \leq a < q-1);\\
        q^{m-t}-q+1, & \text{if } (q-1)q^t \leq \delta \leq q^{t+1}-q+1(1 \leq t \leq m-2);\\
        q-a, & \text{if } aq^{m-1} \leq \delta \leq (a+1)q^{m-1}-1(1 \leq a < q-2);\\
        2, & \text{if } (q-2)q^{m-1} \leq \delta < (q-1)q^{m-1}-q^{\left\lfloor \frac{m-1}{2} \right\rfloor};\\
        (b+1)q^{m-t}-1 & \text{if } \delta=q^t-b(1 \leq t \leq m-1,1 \leq b \leq q-2, q^t-b \geq 3).
    \end{cases}
    \]
\end{lemma}

\begin{lemma}[\cite{10068556}]
    \label{Hermitian dual BCH}
    For $2 \leq \delta \leq n$ and $m \geq 3$, let $I(\delta) \geq 1$ be the integer such that $\{0,1,2,\dots,I(\delta)-1\} \subseteq \textbf{T}_{g^{\perp_H}}$ and $I(\delta) \notin \textbf{T}_{g^{\perp_H}}$. Then
    \[
    I(\delta)=
    \begin{cases}
        (b+1)q_1^{2m-1}-1, & \text{if } \delta=q_1-b(1 \leq b \leq q_1-2);\\
        q_1^{2m-1}-q_1^2+q_1, & \text{if } \delta=q_1^3-q_1^2,\, q_1^3-q_1^2+1;\\
        (b+1)q_1^{2(m-t)+1}-1, & \text{if } \delta=q_1^{2t-1}-b(2 \leq t \leq m-1,\, 1 \leq b \leq q_1^2-2);\\
        q_1^{2(m-t)+1}-s, & \text{if } sq_1^{2t-1} \leq \delta \leq (s+1)q_1^{2t-1}-1(1 \leq t \leq m,\, 1 \leq s \leq q_1-1);\\
        q_1^{2(m-t)+1}-aq_1-s, & \text{if } (aq_1+s)q_1^{2t-1} \leq \delta \leq (aq_1+s+1)q_1^{2t-1}-1\\
        &(1 \leq t \leq m-1,\, 1 \leq a \leq q_1-2,\, 0 \leq s \leq q_1-1);\\
        q_1^{2(m-t)+1}-q_1^2+q_1-s, & \text{if } (q_1^2-q_1+s)q_1^{2t-1} \leq \delta \leq (q_1^2-q_1+s+1)q_1^{2t-1}-1\\
        &(2 \leq t \leq m-1,\, 0 \leq s \leq q_1-2);\\
        q_1^{2(m-t)+1}-q_1^2+1, & \text{if } (q_1^2-1)q_1^{2t-1} \leq \delta \leq q_1^{2t+1}-q_1^2+1(2 \leq t \leq m-1);\\
        q_1^3-q_1^2+1, & \text{if } \delta=q_1^{2m-1}-b(q_1^2-q_1 \leq b \leq q_1^2-2);\\
        (b+1)q_1-1, & \text{if } \delta=q_1^{2m-1}-b(1 \leq b \leq q_1^2-q_1-1).
    \end{cases}
    \]
    Especially, for $m=2$, we have
    \[
    I(\delta)=
    \begin{cases}
        (b+1)q_1^3-1, & \text{if } \delta=q_1-b(1 \le b\le q_1-2);\\
        (b+1)q_1-1, & \text{if } \delta=q_1^3-b(1\le b\le q_1^2-q_1);\\
        q_1^3-q_1^2+q_1, & \text{if } \delta=q_1^3-b(q_1^2-q_1+1\le b\le q_1^2);\\
        q_1^{5-2t}-s, & \text{if } sq_1^{2t-1}\le \delta\le (s+1)q_1^{2t-1}-1(1\le t\le 2,1\le s\le q-1).
    \end{cases}
    \]
\end{lemma}

\section{Tighter Lower Bounds on the Minimum Distance of Several Self‑Dual Cyclic Codes over $\mathbb{F}_{2^s}$ }
This section is devoted to new minimum‑distance estimates for self‑dual cyclic codes. We begin by constructing explicit families of Euclidean self‑dual cyclic codes that attain favorable minimum distances, followed by corresponding results for the Hermitian case. We conclude with several slight refinements of earlier bounds.

\label{3}
\subsection{Euclidean Self-Dual Cyclic Codes for New $\delta$ and $q\ge 4$}

Among the main results of this paper, we obtain the following parameters of Euclidean self-dual cyclic codes.

\begin{theorem}
    \label{4 small delta self-dual cyclic}
    Let $q=2^s$, $s\ge2$ and $m$ be odd. Let $n=q^m-1$ and $\beta$ be an $n$-th primitive root of unity in $\mathbb{F}_{q^m}$. Put $\delta=q^{\frac{m+1}{2}}-2q+1\ge2$. Let $g^\perp(x)=x^{\deg h(x)}h(x^{-1})/h(0)$ be the generator polynomial of $\mathcal{C}_{(q,n,\delta,1,\beta)}^\perp$, where $g(x)$ is the generator polynomial of $\mathcal{C}_{(q,n,\delta,1,\beta)}$ and $h(x)=(x^n-1)/g(x)$. Let $\mathcal{C}'$ denote the cyclic code of length $2n$ over $\mathbb{F}_q$ with generator polynomial $g(x)g^\perp(x)$. Then $\mathcal{C}'$ is a Euclidean self-dual cyclic code with parameters $[2n,n,d]_q$, where
    \begin{itemize}
        \item $d=\min\{2d(\mathcal{C}_{(q,n,\delta,1,\beta)}),d(\mathcal{C}_{(q,n,\delta,1,\beta)}^\perp)\}\ge q^\frac{m+1}{2}+3$ for $m\ge5$.
        \item $d=\min\{2d(\mathcal{C}_{(q,n,\delta,1,\beta)}),d(\mathcal{C}_{(q,n,\delta,1,\beta)}^\perp)\}\ge q^2+4$ for $q\ge8$ and $m=3$.
    \end{itemize}
\end{theorem}

The argument for Theorem~\ref{4 small delta self-dual cyclic} is as follows. The inequality
$d(\mathcal{C}^{\perp}_{(q,n,\delta,1,\beta)}) \ge 2d(\mathcal{C}_{(q,n,\delta,1,\beta)})$
is not generally true. The key point is that a smaller designed distance $\delta$ reduces the minimum distance of $\mathcal{C}_{(q,n,\delta,1,\beta)}$ but increases that of its dual $\mathcal{C}^{\perp}_{(q,n,\delta,1,\beta)}$, which is demonstrated in the following lemma.

\begin{lemma}
    \label{contain}
    Let $q$ be a prime power, $n$ be a positive integer with $\gcd(n,q)=1$. Let $m=\text{ord}_n(q)$ and $\beta$ be an $n$-th primitive root of unity in $\mathbb{F}_{q^m}$. If $2\le\delta_1<\delta_2$,
    then
    \[
    \mathcal{C}_{(q,n,\delta_2,1,\beta)}\subseteq\mathcal{C}_{(q,n,\delta_1,1,\beta)}
    \quad \text{and} \quad
    \mathcal{C}^\perp_{(q,n,\delta_1,1,\beta)}\subseteq\mathcal{C}^\perp_{(q,n,\delta_2,1,\beta)}.
    \]
    Besides, if $q=q_1^2$, then \[\mathcal{C}^{\perp_H}_{(q,n,\delta_1,1,\beta)}\subseteq\mathcal{C}^{\perp_H}_{(q,n,\delta_2,1,\beta)}.\]
\end{lemma}

\begin{proof}
Let $g_1(x)$ and $g_2(x)$ be the generator polynomials of $\mathcal{C}_{(q,n,\delta_1,1,\beta)}$ and $\mathcal{C}_{(q,n,\delta_2,1,\beta)}$, respectively. Since the designed distance of $\mathcal{C}_{(q,n,\delta_1,1,\beta)}$ is less than that of $\mathcal{C}_{(q,n,\delta_2,1,\beta)}$, we have the generator polynomial of $\mathcal{C}_{(q,n,\delta_1,1,\beta)}$ divides that of $\mathcal{C}_{(q,n,\delta_2,1,\beta)}$, which implies
\[
\mathcal{C}_{(q,n,\delta_2,1,\beta)} \subseteq \mathcal{C}_{(q,n,\delta_1,1,\beta)}.
\]

Moreover, the condition
\(
g_2^\perp(\beta^i) = 0 \) implies that $h_2(\beta^{-i}) = 0$, which in turn yields $g_2(\beta^{-i}) \ne 0$.
It follows that $g_1(\beta^{-i}) \ne 0$, which forces $h_1(\beta^{-i}) = 0$, and therefore $g_1^\perp(\beta^i) = 0$. We thus obtain $g_2^\perp(x)$ divides $g_1^\perp(x)$, which means
\[
\mathcal{C}^\perp_{(q,n,\delta_1,1,\beta)} \subseteq \mathcal{C}^\perp_{(q,n,\delta_2,1,\beta)}.
\]

The same inclusion also holds for the Hermitian duals due to \eqref{(1)}.
\end{proof}

Reducing the designed distance $\delta$ does not necessarily lengthen the segment of consecutive zeros in the defining set of $\mathcal{C}^\perp_{(q,n,\delta,1,\beta)}$. However, it generates additional such segments. This observation motivates us to apply a corollary of the BCH bound, which is stated in the following lemma.

\begin{lemma}[~\cite{Blahut2003Algebraic}]
\label{Roos bound}
Let $\mathcal{C}$ be a cyclic code of length $n$ over $\mathbb{F}_q$ with defining set $\mathbf{T}$.
Suppose there exist integers $a$, $b$, $\delta\ge 2$, $s$ and $\gcd(n,b)=1$, such that for every integer $l_1\in \{0,\dots,\delta-s-2\}$, there exist at least $s+1$ integers $l_2\in\{0,\dots,\delta-2\}$ satisfying $a+l_1+l_2b\in \textbf{T}$. Then the minimum distance satisfies $d(\mathcal{C})\ge \delta$.
\end{lemma}

The following theorem shows the additional segment of consecutive zeros in the defining set of $\mathcal{C}_{(q,n,\delta,1,\beta)}^\perp$ with the new designed distance $\delta$.

\begin{theorem}
	\label{4 small delta}
	Let $q=2^s$, $s\ge 2$ and $m$ be odd. Let $n=q^m-1$ and $\beta$ be an $n$-th primitive root of unity in $\mathbb{F}_{q^m}$. Put $\delta=q^{\frac{m+1}{2}}-2q+1\ge2$. Let $g^\perp(x)=x^{\deg h(x)}h(x^{-1})/h(0)$ be the generator polynomial of $\mathcal{C}_{(q,n,\delta,1,\beta)}^\perp$, where $g(x)$ is the generator polynomial of $\mathcal{C}_{(q,n,\delta,1,\beta)}$ and $h(x)=(x^n-1)/g(x)$. Then
	\begin{itemize}
		\item $\{0,1,\dots,q^\frac{m+1}{2}-q,q^\frac{m+1}{2},q^\frac{m+1}{2}+1,\dots,2q^\frac{m+1}{2}-q\}\subseteq \textbf{T}_{g^\perp}$ for $m\ge5$.
		\item $\{0,1,\dots,q^2-q+1,q^2,q^2+1,\dots,2q^2-q+1\}\subseteq \textbf{T}_{g^\perp}$ for $m=3$.
	\end{itemize}
\end{theorem}

\begin{proof}
	We first note that $m \ge 3$ follows directly from $\delta \ge 2$. Consequently, it suffices to consider only the cases $m\ge 5$ and $m=3$.
	
	For the case of $m\ge 5$, by setting $t = \frac{m-1}{2}$, we get
	\[
	(q-1)q^t \le q^{\frac{m+1}{2}} - 2q + 1 \le q^{t+1} - q + 1.
	\]
	By Lemma~\ref{Euclidean dual BCH}, this implies
	\(
	\{0,1,\dots,q^{\frac{m+1}{2}} - q\} \subseteq \mathbf{T}_{g^\perp}.
	\)
	
	Suppose that there exists some
	\[
	j \in \{q^{\frac{m+1}{2}}, \dots, 2q^{\frac{m+1}{2}} - q\}
	\]
	such that $g^\perp(\beta^j) \ne 0$. It forces $h(\beta^{-j}) \ne 0$, which in turn implies $g(\beta^{-j}) = 0$. Hence, there exists an integer
	\[
	i \in \{1, \dots, q^{\frac{m+1}{2}} - 2q\}
	\]
	satisfying
	\(
	i q^l \equiv -j \pmod{n}
	\). If the exponent $l \le \frac{m-1}{2}$, then
	\[
	i q^l \le (q^{\frac{m+1}{2}} - 2q) q^{\frac{m-1}{2}}
	= q^m - 2q^{\frac{m+1}{2}}
	< q^m - 1 = n.
	\]
	Thus from the congruence $i q^l\equiv -j \pmod n$, it follows that $i q^l = n - j$, and consequently
	\[
	n - j \le q^m - 2q^{\frac{m+1}{2}}
	< q^m - 2q^{\frac{m+1}{2}} + q - 1
	\le n - j,
	\]
	which is impossible. Therefore, we must have the exponent $l \ge \frac{m+1}{2}$.
	
	From the congruence $i q^l \equiv -j \pmod{n}$, it follows that
	\(
	j q^{m-l} \equiv -i \pmod{n}.
	\)
	If the exponent $l \ge \frac{m+3}{2}$, then
	\[
	j q^{m-l} \le (2q^{\frac{m+1}{2}} - q) q^{\frac{m-3}{2}}
	= 2q^{m-1} - q^{\frac{m-1}{2}}
	< q^m - 1 = n.
	\]
	Thus from the congruence $j q^{m-l} \equiv -i \pmod{n}$, it follows that $j q^{m-l} = n - i$, and consequently
	\[
	n - i \le 2q^{m-1} - q^{\frac{m-1}{2}}
	< q^m - q^{\frac{m+1}{2}} + 2q - 1
	\le n - i.
	\]
	The second inequality holds if and only if
	\[
	(q-2)q^{m-1} - q^{\frac{m-1}{2}}(q-1) + 2q - 1 > 0,
	\]
	which is equivalent to
	\[
	q^{\frac{m-1}{2}} \big[(q-2)q^{\frac{m-1}{2}} - q + 1\big] + 2q - 1 > 0.
	\]
	Notice that
	\[
	(q-2)q^{\frac{m-1}{2}} - q + 1 \ge 2q - q + 1 = q + 1 > 0,
	\]
	thus the inequality is satisfied, again yielding a contradiction.
	
	In the remaining case of $l=\frac{m+1}{2}$, we have
	\[
	j q^{m-l} \ge q^{\frac{m+1}{2}} \cdot q^{\frac{m-1}{2}} = q^m > q^m - 1 = n
	\]
	and
	\[
	j q^{m-l} \le (2q^{\frac{m+1}{2}} - q) q^{\frac{m-1}{2}}
	= 2q^m - q^{\frac{m+1}{2}}
	< 2(q^m - 1) = 2n.
	\]
	Thus from the congruence $j q^{m-l} \equiv -i \pmod{n}$, it follows that $j q^{m-l} = 2n - i$ and
	\[
	2n - i \le 2q^m - q^{\frac{m+1}{2}}
	< 2q^m - q^{\frac{m+1}{2}} + 2q - 2
	\le 2n - i,
	\]
	which is a contradiction.
	
	For the case of $m=3$, by setting $t = \frac{m-1}{2}$ and $a = q - 2$, we get
	\[
	a q^t \le q^2 - 2q + 1 \le (a+1)q^t - 1.
	\]
	By Lemma~\ref{Euclidean dual BCH}, this gives
	\(
	\{0,1,\dots,q^2 - q + 1\} \subseteq \mathbf{T}_{g^\perp}.
	\)
	The same argument shows
	\[
	\{q^2, q^2+1, \dots, 2q^2 - q\} \subseteq \mathbf{T}_{g^\perp}.
	\]
	It remains to prove that $2q^2 - q + 1 \in \mathbf{T}_{g^\perp}$.
	
	A direct computation yields
	\[
	\{(2q^2 - q + 1) q^l \bmod (q^3 - 1) : 0 \le l \le 2\}
	= \{2q^2 - q + 1,\; q^3 - q^2 + q + 1,\; q^2 + 2q - 1\}.
	\]
	Moreover,
	\[
	\{(2q^2 - q + 1) q^l \bmod (q^3 - 1) : 0 \le l \le 2\}
	\cap
	\{q^3 - q^2 + 2q - 1,\; q^3 - q^2 + 2q,\; \dots,\; q^3 - 2\}
	= \varnothing.
	\]
	This completes the proof of the whole theorem.
\end{proof}

Building on Lemmas~\ref{contain}, \ref{Roos bound} and Theorem~\ref{4 small delta}, we now present the proof of Theorem~\ref{4 small delta self-dual cyclic}.

\begin{proof}[The proof of Theorem \ref{4 small delta self-dual cyclic}]
    For the case of $m \ge 5$, by setting the starting position $a = 0$, the length of consecutive zeros $\mu - s - 1 = q^{\frac{m+1}{2}} - 2q + 1$ and the gap length $b = q^{\frac{m+1}{2}}$, we obtain the number of values of $l_2$ in $\{0,1,\dots,\mu-2\}$ is
    \[
    s + 1 = 2(q+1)
    \]
    and the maximum value of $l_2$ satisfies
    \[
    \mu - 2 \ge q^{\frac{m-1}{2}} \cdot q + 1.
    \]
    By Lemma~\ref{Roos bound}, this yields
    \[
    d(\mathcal{C}^\perp_{(q,n,\delta,1,\beta)}) \ge \mu = q^{\frac{m+1}{2}} + 3.
    \]
    
    Noticing that $d(\mathcal{C}_{(q,n,\delta,1,\beta)}) \ge \delta = q^{\frac{m+1}{2}} - 2q + 1$ and
    \(
    2(q^{\frac{m+1}{2}} - 2q + 1) > q^{\frac{m+1}{2}} + 3,
    \)
    we have
    \[
    d = \min\{2d(\mathcal{C}_{(q,n,\delta,1,\beta)}),\, d(\mathcal{C}^\perp_{(q,n,\delta,1,\beta)})\} \ge q^{\frac{m+1}{2}} + 3.
    \]
    
    For the case of $q\ge 8$ and $m = 3$, by setting $a = 0$, $\mu - s - 1 = q^2 - 2q + 2$ and $b = q^2$, we obtain the number of values of $l_2$ in $\{0,1,\dots,\mu-2\}$ is
    \[
    s + 1 = 2(q+1)
    \]
    and the maximum value of $l_2$ satisfies
    \[
    \mu - 2 \ge q^2 + 1.
    \]
    By Lemma~\ref{Roos bound}, this yields
    \[
    d(\mathcal{C}^\perp_{(q,n,\delta,1,\beta)}) \ge \mu = q^2 + 4.
    \]
    
    Noticing that $d(\mathcal{C}_{(q,n,\delta,1,\beta)}) \ge \delta = q^2 - 2q + 1$ and
    \(
    2\delta = 2(q^2 - 2q + 1) \ge q^2 + 4,
    \)
    we have
    \[
    d = \min\{2d(\mathcal{C}_{(q,n,\delta,1,\beta)}),\, d(\mathcal{C}^\perp_{(q,n,\delta,1,\beta)})\} \ge q^2 + 4.
    \]
    This completes the proof of the whole theorem.
\end{proof}

\subsection{Euclidean Self-Dual Cyclic Codes for New $\delta$ and $q=2$}

In this subsection, we discuss the case $q = 2$. While the results in~\cite{wang2026constructionsselfdualbinarycyclic} yield better lower bounds for $m \ge 11$, our approach provides more general results for $m \ge 5$.
\begin{theorem}
    \label{2 small delta}
    Let $q=2$ and $m$ be odd. Let $n=2^m-1$ and $\beta$ be an $n$-th primitive root of unity in $\mathbb{F}_{q^m}$. Put $\delta=2^{\frac{m+1}{2}}-3\ge2$. Let $g^\perp(x)=x^{\deg h(x)}h(x^{-1})/h(0)$ be the generator polynomial of $\mathcal{C}_{(2,n,\delta,1,\beta)}^\perp$, where $g(x)$ is the generator polynomial of $\mathcal{C}_{(2,n,\delta,1,\beta)}$ and $h(x)=(x^n-1)/g(x)$. Then \[\{0,1,\dots,2^\frac{m+1}{2}-2,2^\frac{m+1}{2},2^\frac{m+1}{2}+1,\dots,2^\frac{m+3}{2}-4\}\subseteq \textbf{T}_{g^\perp} \text{~for~} m\ge 5.\]
    Especially, for the case of $m=5$, we have $2^\frac{m+3}{2}-3\in \textbf{T}_{g^\perp}$
\end{theorem}

\begin{proof}
We first note that $m \ge 5$ follows directly from $\delta \ge 2$. By setting $t = \frac{m-1}{2}$, according to the inequality
\(
2^t \le 2^{\frac{m+1}{2}} - 3 \le 2^{t+1} - 1
\) and Lemma~\ref{Euclidean dual BCH}, we have
\[
\{0,1,\dots,2^{\frac{m+1}{2}} - 2\} \subseteq \mathbf{T}_{g^\perp}.
\]

Suppose that there exists some
\[
j \in \{2^{\frac{m+1}{2}}, 2^{\frac{m+1}{2}} + 1, \dots, 2^{\frac{m+3}{2}} - 4\}
\]
such that $g^\perp(\beta^j) \ne 0$. It forces $h(\beta^{-j}) \ne 0$, which in turn implies $g(\beta^{-j}) = 0$. Hence, there exists an integer
\[
i \in \{1, 2, \dots, 2^{\frac{m+1}{2}} - 4\}
\]
satisfying
\(
i \cdot 2^l \equiv -j \pmod{n}
\). If the exponent $l \le \frac{m-1}{2}$, then
\[
i \cdot 2^l \le (2^{\frac{m+1}{2}} - 4) \cdot 2^{\frac{m-1}{2}}
= 2^m - 2^{\frac{m+3}{2}}
< 2^m - 1 = n.
\]
From the congruence $i \cdot 2^l \equiv -j \pmod{n}$, it follows that $i \cdot 2^l = n - j$, and consequently
\[
n - j \le 2^m - 2^{\frac{m+3}{2}}
< 2^m - 2^{\frac{m+3}{2}} + 3
\le n - j,
\]
which is impossible. Therefore, we must have the exponent $l \ge \frac{m+1}{2}$.

From the congruence $i \cdot 2^l \equiv -j \pmod{n}$, it follows that
\(
j \cdot 2^{m-l} \equiv -i \pmod{n}
\). If the exponent $l \ge \frac{m+3}{2}$, then
\[
j \cdot 2^{m-l} \le (2^{\frac{m+3}{2}} - 4) \cdot 2^{\frac{m-3}{2}}
= 2^m - 2^{\frac{m+1}{2}}
< 2^m - 1 = n.
\]
Thus from the congruence $j \cdot 2^{m-l} \equiv -i \pmod{n}$, it follows that $j \cdot 2^{m-l} = n - i$, and consequently
\[
n - i \le 2^m - 2^{\frac{m+1}{2}}
< 2^m - 2^{\frac{m+1}{2}} + 3
\le n - i,
\]
a contradiction. In the remaining case of $l=\frac{m+1}{2}$, we have
\[
j \cdot 2^{m-l} \ge 2^{\frac{m+1}{2}} \cdot 2^{\frac{m-1}{2}}
= 2^m > 2^m - 1 = n
\]
and
\[
j \cdot 2^{m-l} \le (2^{\frac{m+3}{2}} - 4) \cdot 2^{\frac{m-1}{2}}
= 2^{m+1} - 2^{\frac{m+3}{2}}
< 2(2^m - 1) = 2n.
\]
Thus from the congruence $j \cdot 2^{m-l} \equiv -i \pmod{n}$, it follows that $j \cdot 2^{m-l} = 2n - i$ and
\[
2n - i \le 2^{m+1} - 2^{\frac{m+3}{2}}
< 2^{m+1} - 2^{\frac{m+1}{2}} + 2
\le 2n - i,
\]
which is a contradiction. It remains to determine whether the element $2^{\frac{m+3}{2}} - 3$ belongs to $\mathbf{T}_{g^\perp}$.

For the case of $m \ge 7$, notice that
\(
1 \le 3 \cdot 2^{\frac{m-3}{2}} - 1 \le 2^{\frac{m+1}{2}} - 4
\)
and
\[
(3 \cdot 2^{\frac{m-3}{2}} - 1) \cdot 2^{\frac{m+3}{2}}
\equiv 2^m - 2^{\frac{m+3}{2}} + 2 \pmod{n}.
\]
Thus the element $2^m - 2^{\frac{m+3}{2}} + 2 \in \mathbf{T}_g$, which implies
\(
2^{\frac{m+3}{2}} - 3 \notin \mathbf{T}_{g^\perp}.
\)

For the case of $m = 5$, a direct computation yields
\(
2^{\frac{m+3}{2}} - 3 = 13 \in \mathbf{T}_{g^\perp}
\). It is worthy noticing that
\(
1 \le 2^{\frac{m-1}{2}} - 1 \le 2^{\frac{m+1}{2}} - 4
\)
and
\[
(2^{\frac{m-1}{2}} - 1) \cdot 2^{\frac{m+3}{2}}
\equiv 2^m - 2^{\frac{m+3}{2}} + 1 \pmod{n}.
\]
Hence, the element $2^m - 2^{\frac{m+3}{2}} + 1 \in \mathbf{T}_g$, which implies
\(
2^{\frac{m+3}{2}} - 2 \notin \mathbf{T}_{g^\perp}.
\)
\end{proof}

By Lemma~\ref{Roos bound}, we obtain Euclidean self-dual cyclic codes with the following parameters.

\begin{theorem}
    \label{2 small delta self-dual cyclic}
    Let $q=2$ and $m$ be odd. Let $n=2^m-1$ and $\beta$ be an $n$-th primitive root of unity in $\mathbb{F}_{2^m}$. Put $\delta=2^{\frac{m+1}{2}}-3\ge2$. Let $g^\perp(x)=x^{\deg h(x)}h(x^{-1})/h(0)$ be the generator polynomial of $\mathcal{C}_{(2,n,\delta,1,\beta)}^\perp$, where $g(x)$ is the generator polynomial of $\mathcal{C}_{(2,n,\delta,1,\beta)}$ and $h(x)=(x^n-1)/g(x)$. Let $\mathcal{C}'$ denote the cyclic code of length $2n$ over $\mathbb{F}_q$ with generator polynomial $g(x)g^\perp(x)$. Then $\mathcal{C}'$ is a Euclidean self-dual cyclic code with parameters $[2n,n,d]_q$, where \[d=\min\{2d(\mathcal{C}_{(2,n,\delta,1,\beta)}),d(\mathcal{C}_{(2,n,\delta,1,\beta)}^\perp)\}\ge2^\frac{m+1}{2}.\]
    In particular, the minimum distance $d\ge 2^\frac{m+1}{2}+2=10$ for the case of $m=5$.
\end{theorem}

\begin{proof}
    For the case of $m \ge 7$, the lower bounds obtained by the Roos bound coincide with those from the BCH bound.

For the case of $m = 5$, by setting the starting position $a = 0$, the length of consecutive zeros $\mu - s - 1 = 2^{\frac{m+1}{2}} - 2q + 1 = 5$ and the gap length $b = 2^{\frac{m+1}{2}} = 8$, we obtain the number of values of $l_2$ in $\{0,1,\dots,\mu-2\}$ is
\[
s + 1 = 2q + 1 = 5
\]
and the maximum value of $l_2$ satisfies
\[
\mu - 2 \ge 2q^{\frac{m-1}{2}} = 8.
\]
By Lemma~\ref{Roos bound}, this yields
\[
d(\mathcal{C}^\perp_{(2,n,\delta,1,\beta)}) \ge \mu = 10.
\]
Noticing that $d(\mathcal{C}_{(2,n,\delta,1,\beta)}) \ge \delta = 5$, we have
\[
d = \min\{2d(\mathcal{C}_{(2,n,\delta,1,\beta)}),\, d(\mathcal{C}^\perp_{(2,n,\delta,1,\beta)})\} \ge 10.
\]
This completes the proof.
\end{proof}

The following remarks concerning the Euclidean self-dual cyclic codes $\mathcal{C}'$ in Theorems~\ref{4 small delta self-dual cyclic} and \ref{2 small delta self-dual cyclic} are in order.

\begin{itemize}
    \item The self-dual cyclic code $\mathcal{C}'$ in Theorem~\ref{4 small delta self-dual cyclic} has parameters
    \[
    [2(q^m-1),\, q^m-1,\, d \ge q^{\frac{m+1}{2}} + 3]_q
    \]
    for the case of $m \ge 5$. For the case of $q \ge 8$ and $m = 3$, it has parameters
    \[
    [2(q^3-1),\, q^3-1,\, d \ge q^2 + 4]_q.
    \]
    For the case of $q = 4$ and $m = 3$, due to $2\delta = 18 < 20 = q^2 + 4$, it has parameters
    \[
    [126,\, 63,\, d \ge 18]_4,
    \]
    whose lower bound on the minimum distance improves upon that of the corresponding code in \cite{10856234}.

    \item The code $\mathcal{C}'$ in Theorem~\ref{2 small delta self-dual cyclic} has parameters
    \[
    [2(2^m-1),\, 2^m-1,\, d \ge 2^{\frac{m+1}{2}}]_2
    \]
    for the case of $m \ge 7$. This lower bound coincides with that in Theorem~\ref{new Euclidean odd}. For the case of $m = 5$, it has parameters
    \[
    [62,\, 31,\, d \ge 10]_2,
    \]
    whose lower bound on the minimum distance is stronger than that of the corresponding code in \cite{10856234}.

    \item In both theorems, the lower bounds on the minimum distance of the Euclidean self-dual cyclic codes $\mathcal{C}'$ are square-root lower bounds.

    \item It should be noted that a better lower bound does not necessarily imply a larger actual minimum distance. This is because the lower bound on the minimum distance of $\mathcal{C}_{(q,n,\delta,1,\beta)}^\perp$ may be too weak when $\delta = q^{\frac{m+1}{2}} - q + 1$.
\end{itemize}

Table~\ref{design distance} lists the parameters of Euclidean self-dual cyclic codes for various designed distances, where the dash ``——'' indicates that the corresponding parameter is not applicable.

\begin{table}[H]
  \centering
  \small
  \caption{Parameters of Euclidean Self-dual Cyclic Codes with Various Designed Distances}
  \label{design distance}
  \begin{tabular}{|c|c|c|c|c|c|c|c|c|}
    \hline
    $q$ & $m$ & $n$ & $\delta$ & $d(\mathcal{C}_{(q,n,\delta,1,\beta)})$ & $d(\mathcal{C}_{(q,n,\delta,1,\beta)}^\perp)$ & $d(\mathcal{C}')$ & $d(\mathcal{C}')$ in~\cite{10856234} & $d(\mathcal{C}')$ in Theorems~\ref{4 small delta self-dual cyclic} and \ref{2 small delta self-dual cyclic} \\
    \hline
    2 & 5 & 31 & 7 & 7 & 8 & 8 & $\ge 7$ & —— \\
    \hline
    2 & 5 & 31 & 5 & 5 & 12 & 10 & —— & $\ge 10$ \\
    \hline
    2 & 7 & 127 & 15 & 15 & 28 & 28 & $\ge 15$ & —— \\
    \hline
    2 & 7 & 127 & 13 & 13 & 32 & 26 & —— & $\ge 16$ \\
    \hline
    4 & 3 & 63 & 13 & 13 & 16 & 16 & $\ge 13$ & —— \\
    \hline
    4 & 3 & 63 & 9 & 9 & 21 & 18 & —— & $\ge 18$ \\
    \hline
  \end{tabular}
\end{table}

\subsection{Hermitian Self-Dual Cyclic Code for New $\delta$}

% 版本六：更简洁，用 "also" 隐含 "another"
The following parameters of Hermitian self-dual cyclic codes also constitute a main result of this paper.

\begin{theorem}
	\label{even Hermitian low delta}
	Let $q_1$ be a power of $2$, $q=q_1^2$ and $m$ be even. Let $n=q^m-1$ and $\beta$ be an $n$-th primitive root of unity in $\mathbb{F}_{q^m}$. Put $\delta=q_1^{m+1}-2q_1^2+1\ge 2$. Let $g^{\perp_H}(x)$ be the generator polynomial of $\mathcal{C}_{(q,n,\delta,1,\beta)}^{\perp_H}$. Let $\mathcal{C}'$ denote the cyclic code of length $2n$ over $\mathbb{F}_q$ with generator polynomial $g(x)g^{\perp_H}(x)$, where $g(x)$ is the generator polynomial of $\mathcal{C}_{(q,n,\delta,1,\beta)}$. Then
	\begin{itemize}
		\item $d=\min\{2d(\mathcal{C}_{(q,n,\delta,1,\beta)}),d(\mathcal{C}_{(q,n,\delta,1,\beta)}^{\perp_H})\}\ge q_1^{m+1}+3$ for $m\ge4$.
		\item $d=\min\{2d(\mathcal{C}_{(q,n,\delta,1,\beta)}),d(\mathcal{C}_{(q,n,\delta,1,\beta)}^{\perp_H})\}\ge q_1^3+2q_1+4$ for $q_1\ge8$ and $m=2$.
		\item $d=\min\{2d(\mathcal{C}_{(q,n,\delta,1,\beta)}),d(\mathcal{C}_{(q,n,\delta,1,\beta)}^{\perp_H})\}\ge66$ for $q_1=4$ and $m=2$.
	\end{itemize}
\end{theorem}

Similarly, we first establish the following result concerning the defining sets.

\begin{theorem}
    \label{Hermitian low delta}
    Let $q_1$ be a power of $2$, $q=q_1^2$ and $m$ be even. Let $n=q^m-1$ and $\beta$ be an $n$-th primitive root of unity in $\mathbb{F}_{q^m}$. Put $\delta=q_1^{m+1}-2q_1^2+1\ge 2$. Let $g^{\perp_H}(x)$ be the generator polynomial of $\mathcal{C}_{(q,n,\delta,1,\beta)}^{\perp_H}$. Then
    \begin{itemize}
        \item $\{0,1,\dots,q_1^{m+1}-q_1^2,q_1^{m+1},q_1^{m+1}+1\dots,2q_1^{m+1}-q_1^2\}\subseteq\textbf{T}_{g^{\perp_H}}$ for $m\ge 4$.
        \item $\{0,1,\dots,q_1^3-q_1^2+2q_1-1,q_1^3,q_1^3+1,\dots,2q_1^3-q_1^2+2q_1-1\}\subseteq\textbf{T}_{g^{\perp_H}}$ for $q_1\ge4$ and $m=2$.
    \end{itemize}
\end{theorem}

\begin{proof}
For the case of $m\ge 4$, by setting $t = \frac{m}{2}$, we have
\(
(q_1^2 - 1)q_1^{2t-1} \le q_1^{m+1} - 2q_1^2 + 1 \le q_1^{2t+1} - q_1^2 + 1
\). By Lemma~\ref{Hermitian dual BCH}, this implies
\[
\{0,1,\dots,q_1^{m+1} - q_1^2\} \subseteq \mathbf{T}_{g^{\perp_H}}.
\]

Suppose that there exists some
\[
j \in \{q_1^{m+1}, q_1^{m+1} + 1, \dots, 2q_1^{m+1} - q_1^2\}
\]
such that $g^{\perp_H}(\beta^j) \ne 0$. Then there exists an integer
\[
i \in \{1, \dots, q_1^{m+1} - 2q_1^2\}
\]
satisfying
\(
i q_1^{2l+1} \equiv -j \pmod{n}
\). If $l \le \frac{m-2}{2}$, then
\[
i q_1^{2l+1} \le (q_1^{m+1} - 2q_1^2) \cdot q_1^{m-1}
= q_1^{2m} - 2q_1^{m+1}
< q_1^{2m} - 1 = n.
\]
Thus from the congruence $i q_1^{2l+1} \equiv -j \pmod{n}$, it follows that $i q_1^{2l+1} = n - j$, and consequently
\[
n - j \le q_1^{2m} - 2q_1^{m+1}
< q_1^{2m} - 2q_1^{m+1} + q_1^2 - 1
\le n - j,
\]
which is impossible. Therefore, we must have $l \ge \frac{m}{2}$.

From the congruence $i q_1^{2l+1} \equiv -j \pmod{n}$, it follows that
\(
j q_1^{2m-2l-1} \equiv -i \pmod{n}
\). If $l \ge \frac{m+2}{2}$, then
\[
j q_1^{2m-2l-1} \le (2q_1^{m+1} - q_1^2) \cdot q_1^{m-3}
= 2q_1^{2m-2} - q_1^{m-1}
< q_1^{2m} - 1 = n.
\]
Thus from the congruence $j q_1^{2m-2l-1} \equiv -i \pmod{n}$, it follows that $j q_1^{2m-2l-1} = n - i$, and consequently
\[
n - i \le 2q_1^{2m-2} - q_1^{m-1}
< q_1^{2m} - q_1^{m+1} + 2q_1^2 - 1
\le n - i.
\]
The second inequality holds if and only if
\[
(q_1^2 - 2)q_1^{2m-2} - q_1^{m-1}(q_1^2 - 1) + 2q_1^2 - 1 > 0,
\]
which is equivalent to
\[
q_1^{m-1}\big[(q_1^2 - 2)q_1^{m-1} - q_1^2 + 1\big] + 2q_1^2 - 1 > 0.
\]
Notice that
\[
(q_1^2 - 2)q_1^{m-1} - q_1^2 + 1 \ge 2q_1^3 - q_1^2 + 1 > 0,
\]
thus the inequality is satisfied, again yielding a contradiction. In the remaining case of $l=\frac{m}{2}$, we have
\[
j q_1^{2m-2l-1} \ge q_1^{m+1} \cdot q_1^{m-1} = q_1^{2m} > q_1^{2m} - 1 = n
\]
and
\[
j q_1^{2m-2l-1} \le (2q_1^{m+1} - q_1^2) \cdot q_1^{m-1}
= 2q_1^{2m} - q_1^{m+1}
< 2(q_1^{2m} - 1) = 2n.
\]
Thus from the congruence $j q_1^{2m-2l-1} \equiv -i \pmod{n}$, it follows that $j q_1^{2m-2l-1} = 2n - i$ and
\[
2n - i \le 2q_1^{2m} - q_1^{m+1}
< 2q_1^{2m} - q_1^{m+1} + 2q_1^2 - 2
\le 2n - i,
\]
which is a contradiction. It remains to consider the case of $q_1\ge 4$ and $m = 2$.

Let $\delta_{\max} = q_1^3 - q_1^2 + 1$ and $G(x)$ be the generator polynomial of $\mathcal{C}_{(q_1, n, \delta_{\max}, 1, \beta)}$, where $n = q_1^4 - 1$. By setting $b=q_1^2-1$, Lemma~\ref{Hermitian dual BCH} gives $\{0,1,\dots,q_1^3 - q_1^2 + q_1 - 1\} \subseteq \mathbf{T}_{g^{\perp}}$

By Lemma~\ref{contain}, we have
\[
g(x) \mid G(x) \mid G^{\perp_H}(x) \mid g^{\perp_H}(x),
\]
which implies
\[
\{0,1,\dots,q_1^3 - q_1^2 + q_1 - 1\} \subseteq \mathbf{T}_{g^{\perp_H}}.
\]

Suppose that there exists some
\[
j \in \{q_1^3 - q_1^2 + q_1, q_1^3 - q_1^2 + q_1 + 1, \dots, q_1^3 - q_1^2 + 2q_1 - 1\}
\]
such that $g^{\perp_H}(\beta^j) \ne 0$. Then there exists an integer
\[
i \in \{1, \dots, q_1^3 - 2q_1^2\}
\]
satisfying
\(
i q_1^{2l+1} \equiv -j \pmod{n}
\) with $0\le l \le 1$. For the case of $l = 0$, notice that
\[
i q_1 \le (q_1^3 - 2q_1^2) \cdot q_1 = q_1^4 - 2q_1^3 < q_1^4 - 1 = n.
\]
Thus from the congruence $i q_1\equiv -j \pmod n$, it follows that $i q_1 = n - j$ and
\[
n - j \le q_1^4 - 2q_1^3
< q_1^4 - q_1^3 + q_1^2 - 2q_1 + 1
\le n - j,
\]
which is impossible.

From the congruence $i q_1^{2l+1}\equiv -j \pmod n$, we obtain $j q_1^{3-2l}\equiv -i\pmod n$. For the case of $l = 1$, notice that
\[
j q_1 \le (q_1^3 - q_1^2 + 2q_1 - 1) \cdot q_1
= q_1^4 - q_1^3 + 2q_1^2 - q_1
< q_1^4 - 1 = n.
\]
Thus from the congruence $j q_1\equiv -i \pmod n$, it follows that $j q_1 = n - i$, and
\[
n - i \le q_1^4 - q_1^3 + 2q_1^2 - q_1
< q_1^4 - q_1^3 + 2q_1^2 - 1
\le n - i,
\]
which is a contradiction. Noticing that $q_1^3 - 2q_1^2 - 1 \in \mathbf{T}_g$ and
\[
q_1(q_1^3 - 2q_1^2 - 1) \equiv q_1^4 - q_1^2 + q_1 - 3 \pmod{n},
\]
we have $q_1^4 - q_1^2 + q_1 - 3 \in \mathbf{T}_g$. From the congruence
\[
q_1^3 - q_1^2 + 2q_1 \equiv n - q_1(q_1^4 - q_1^2 + q_1 - 3) \pmod{n},
\]
we obtain the element $q_1^3 - q_1^2 + 2q_1$ not belongs to $\mathbf{T}_{g^{\perp_H}}$. We next prove that
\[
\{q_1^3, q_1^3 + 1, \dots, 2q_1^3 - q_1^2 + 2q_1 - 1\} \subseteq \mathbf{T}_{g^{\perp_H}}.
\]

Suppose that there exists some
\[
j \in \{q_1^3, q_1^3 + 1, \dots, 2q_1^3 - q_1^2 + 2q_1 - 1\}
\]
such that $g^{\perp_H}(\beta^j) \ne 0$. Then there exists an integer
\[
i \in \{1, \dots, q_1^3 - 2q_1^2\}
\]
satisfying
\(
i q_1^{2l+1} \equiv -j \pmod{n}
\) with $0\le l \le 1$. For the case of $l = 0$, notice that
\[
i q_1 \le (q_1^3 - 2q_1^2) \cdot q_1 = q_1^4 - 2q_1^3 < q_1^4 - 1 = n.
\]
Thus from the congruence $i q_1\equiv -j\pmod n$, it follows that $i q_1 = n - j$ and
\[
n - j \le q_1^4 - 2q_1^3
< q_1^4 - 2q_1^3 + q_1^2 - 2q_1 + 1
\le n - j,
\]
a contradiction. From the congruence $i q_1^{2l+1}\equiv -j \pmod n$, we obtain $j q_1^{3-2l}\equiv -i\pmod n$. For the case of $l = 1$, notice that
\[
j q_1 \ge q_1^3 \cdot q_1 = q_1^4 > q_1^4 - 1 = n,
\]
and
\[
j q_1 \le (2q_1^3 - q_1^2 + 2q_1 - 1) \cdot q_1
= 2q_1^4 - q_1^3 + 2q_1^2 - q_1
< 2(q_1^4 - 1) = 2n.
\]
Thus from the congruence $j q_1\equiv -i \pmod n$, it follows that $j q_1 = 2n - i$ and
\[
2n - i \le 2q_1^4 - q_1^3 + 2q_1^2 - q_1
< 2q_1^4 - q_1^3 + 2q_1^2 - 2
\le 2n - i,
\]
which is impossible. This completes the proof.
\end{proof}

We are now in a position to  prove Theorem~\ref{even Hermitian low delta} by  applying Lemma~\ref{Roos bound}.

\begin{proof}[The proof of Theorem \ref{even Hermitian low delta}]
For the case of $m \ge 4$, by setting the starting position $a = 0$, the length of consecutive zeros
\[
\mu - s - 1 = q_1^{m+1} - 2q_1^2 + 1,
\]
and the gap length $b = q_1^{m+1}$, we obtain the number of values of $l_2$ in $\{0,1,\dots,\mu-2\}$ is
\[
s + 1 = 2(q_1^2 + 1)
\]
and the maximum value of $l_2$ satisfies
\[
\mu - 2 \ge q_1^{m-1} \cdot q_1^2 + 1.
\]
By Lemma~\ref{Roos bound}, this yields
\[
d(\mathcal{C}^{\perp_H}_{(q,n,\delta,1,\beta)}) \ge \mu = q_1^{m+1} + 3.
\]

Noticing that $d(\mathcal{C}_{(q,n,\delta,1,\beta)}) \ge \delta = q_1^{m+1} - 2q_1^2 + 1$ and
\(
2\delta = 2(q_1^{m+1} - 2q_1^2 + 1) > q_1^{m+1} + 3
\), we have
\[
d = \min\{2d(\mathcal{C}_{(q,n,\delta,1,\beta)}),\, d(\mathcal{C}_{(q,n,\delta,1,\beta)}^{\perp_H})\} \ge q_1^{m+1} + 3.
\]

For the case of $m = 2$, by setting the starting position $a = 0$, the length of consecutive zeros
\[
\mu - s - 1 = q_1^3 - 2q_1^2 + 2q_1 - 2
\]
and the gap length $b = q_1^3$, we obtain the number of values of $l_2$ in $\{0,1,\dots,\mu-2\}$ is
\[
s + 1 = 2(q_1^2 + 3)
\]
and the maximum value of $l_2$ satisfies
\[
\mu - 2 \ge q_1(q_1^2 + 2) + 1.
\]
By Lemma~\ref{Roos bound}, this yields
\[
d(\mathcal{C}^{\perp_H}_{(q,n,\delta,1,\beta)}) \ge \mu = q_1^3 + 2q_1 + 4.
\]

For the case of $q_1 \ge 8$, it is worthy noticing that $d(\mathcal{C}_{(q,n,\delta,1,\beta)}) \ge \delta = q_1^3 - 2q_1^2 + 1$ and
\(
2\delta = 2(q_1^3 - 2q_1^2 + 1) > q_1^3 + 2q_1 + 4
\). It follows that
\[
d = \min\{2d(\mathcal{C}_{(q,n,\delta,1,\beta)}),\, d(\mathcal{C}_{(q,n,\delta,1,\beta)}^{\perp_H})\} \ge q_1^3 + 2q_1 + 4.
\]
For the case of $q_1 = 4$, it is worthy noticing that $\delta = q_1^3 - 2q_1^2 + 1 = 33$, $q_1^3 + 2q_1 + 4 = 76$ and $2\delta = 66 < 76$. It follows that
\[
d = \min\{2d(\mathcal{C}_{(q,n,\delta,1,\beta)}),\, d(\mathcal{C}_{(q,n,\delta,1,\beta)}^{\perp_H})\} \ge 66.
\]
This completes the whole proof of theorem.
\end{proof}

The following remarks concerning the Hermitian self-dual cyclic code $\mathcal{C}'$ in Theorem~\ref{even Hermitian low delta} are in order.

\begin{itemize}
    \item The Hermitian self-dual cyclic code $\mathcal{C}'$ has parameters
    \[
    [2(q_1^{2m} - 1),\, q_1^{2m} - 1,\, d \ge q_1^{m+1} + 3]_{q_1^2}
    \]
    for the case of $m \ge 4$. For the case of $q_1 \ge 8$ and $m = 2$, it has parameters
    \[
    [2(q_1^4 - 1),\, q_1^4 - 1,\, d \ge q_1^3 + 2q_1 + 4]_{q_1^2}.
    \]
    In particular, for the case of $q_1 = 4$ and $m = 2$, it has parameters
    \[
    [510,\, 255,\, d \ge 66]_{16}.
    \]
    \item The lower bound on the minimum distance of the Hermitian self-dual cyclic code $\mathcal{C}'$ is a square-root lower bound.
    \item It should be noted that a better lower bound does not necessarily imply a larger actual minimum distance.
\end{itemize}

\subsection{Refinement of Lower Bounds for the Previous $\delta$}

Applying Lemmas~\ref{Hermitian dual BCH} and~\ref{Euclidean dual BCH} yields improved lower bounds on the minimum distances of the Euclidean (or Hermitian) dual codes, which in turn refines the lower bounds for the corresponding self-dual cyclic codes. The following theorem gives a slight refinement on the existing lower bounds of Euclidean self-dual cyclic codes with previous $\delta$.

\begin{theorem}
    \label{new Euclidean odd}
    The Euclidean self-dual cyclic code $\mathcal{C}'$ in Lemma \ref{odd Euclidean self-dual} has parameters \[[2(q^m-1),q^m-1,d\ge q^\frac{m+1}{2}-q+2]_q.\]
\end{theorem}

\begin{proof}
    By setting $t = \frac{m-1}{2}$, we obtain $\delta = q^{t+1} - q + 1 = q^{\frac{m+1}{2}} - q + 1$. Lemma~\ref{Euclidean dual BCH} implies
\[
\{0,1,\dots,q^{\frac{m+1}{2}} - q\} \subseteq \mathbf{T}_{g^\perp}.
\]
Applying the BCH bound then yields
\[
d(\mathcal{C}^\perp_{(q,n,\delta,1,\beta)}) \ge q^{\frac{m+1}{2}} - q + 2.
\]
Consequently, we have
\[
d(\mathcal{C}') = \min\{2d(\mathcal{C}_{(q,n,\delta,1,\beta)}),\, d(\mathcal{C}^\perp_{(q,n,\delta,1,\beta)})\} \ge q^{\frac{m+1}{2}} - q + 2.
\]
This completes the proof.
\end{proof}

Table~\ref{odd} lists the parameters of several Euclidean self-dual cyclic codes.

\begin{table}[htbp]
  \centering 
  \caption{Parameters of Euclidean self-dual cyclic codes for odd $m$}
  \label{odd}
  \begin{tabular}{|c|c|c|c|c|c|c|c|c|}
    \hline 
    $q$ & $m$ & $n$ & $\delta$ & $d(\mathcal{C}_{(q,n,\delta,1,\beta)})$ & $d(\mathcal{C}_{(q,n,\delta,1,\beta)}^\perp)$ & $d(\mathcal{C}')$ & $d(\mathcal{C}')$ in \cite{10856234} & $d(\mathcal{C}')$ in Theorem \ref{new Euclidean odd} \\
    \hline
    2 & 3 & 7 & 3 & $3$ & $4$ & $4$ & $\ge 3$ & $\ge 4$ \\
    \hline
    2 & 5 & 31 & 7 & $7$ & $8$ & $8$ & $\ge 7$ & $\ge 8$ \\
    \hline
    4 & 3 & 63 & 13 & $13$ & $16$ & $16$ & $\ge 13$ & $\ge 14$ \\
    \hline
  \end{tabular}
\end{table}

Next, we establish lower bounds on the minimum distances of Hermitian self-dual cyclic codes for even $m$. The following theorem gives a slight refinement on the existing lower bounds of Hermitian self-dual cyclic codes with previous $\delta$.
\begin{theorem}
    \label{new Hermitian self-dual}
    The Hermitian self-dual cyclic code $\mathcal{C}'$ in Theorem \ref{even Hermitian self-dual} has parameters
    \begin{itemize}
        \item $[2(q^m-1),q^m-1,d\ge q_1^{m+1}-q_1^2+2]_q$ for $m\ge 4$.
        \item $[2(q^2-1),q^2-1,d\ge q_1^3-q_1^2+q_1+1]_q$ for $m=2$.
    \end{itemize}
\end{theorem}

\begin{proof}
For the case of $m \ge 4$, by setting $t = \frac{m}{2}$, we obtain $\delta = q_1^{2t+1} - q_1^2 + 1 = q_1^{m+1} - q_1^2 + 1$. Lemma~\ref{Hermitian dual BCH} implies
\[
\{0,1,\dots,q_1^{m+1} - q_1^2\} \subseteq \mathbf{T}_{g^{\perp_H}}.
\]
Applying the BCH bound then yields
\[
d(\mathcal{C}^{\perp_H}_{(q,n,\delta,1,\beta)}) \ge q_1^{m+1} - q_1^2 + 2.
\]

For the case of $m = 2$, by setting $b = q_1^2 - 1$, we obtain $\delta = q_1^3 - b = q_1^3 - q_1^2 + 1$. Lemma~\ref{Hermitian dual BCH} gives
\[
\{0,1,\dots,q_1^3 - q_1^2 + q_1 - 1\} \subseteq \mathbf{T}_{g^{\perp_H}}.
\]
The BCH bound then implies
\[
d(\mathcal{C}^{\perp_H}_{(q,n,\delta,1,\beta)}) \ge q_1^3 - q_1^2 + q_1 + 1.
\]

The desired lower bounds then follow immediately from
\(
d(\mathcal{C}') = \min\{2d(\mathcal{C}_{(q,n,\delta,1,\beta)}),\, d(\mathcal{C}^{\perp_H}_{(q,n,\delta,1,\beta)})\}.
\)
\end{proof}

\section{Exact Parameters of Two Classes of Self‑Dual Cyclic Codes over $\mathbb{F}_{2^s}$}
\label{4}
In this section, we determine the minimum distance of the cyclic codes $\mathcal{C}_{(q,n,\delta,1,\beta)}$ and derive the exact parameters for the constructed self-dual cyclic codes using the lower bounds on the minimum distances of the dual codes. The following lemma will be required.

\begin{lemma}[\cite{li2017lcd}]
    \label{weight}
    Let $\mathcal{C}_{(q,n,\delta,b,\beta)}$ be the BCH code over $\mathbb{F}_q$ of length $n$ with designed distance $\delta$. Then its minimum distance satisfies $d=\delta$ if $\delta$ divides $\gcd(n,b-1)$.
\end{lemma}

We first determine the parameters of Euclidean self‑dual cyclic codes for even $m$.
\begin{theorem}
    \label{exact Euclidean}
    The Euclidean self-dual cyclic code $\mathcal{C}'$ in Theorem \ref{even Euclidean self-dual} has parameters \[[2(q^m-1),q^m-1,2(q^\frac{m}{2}-1)]_q.\]
\end{theorem}

\begin{proof}
It is worthy noticing that
\[
\gcd(n, b-1) = n = q^m - 1 = (q^{\frac{m}{2}} - 1)(q^{\frac{m}{2}} + 1)
\]
and $\delta = q^{\frac{m}{2}} - 1$, we have $\delta \mid \gcd(n, b-1)$. By Lemma~\ref{weight}, this implies
\[
d(\mathcal{C}_{(q,n,\delta,1,\beta)}) = q^{\frac{m}{2}} - 1.
\]

By setting $t = \frac{m}{2}$ and $b = 1$, we have $\delta = q^t - b = q^{\frac{m}{2}} - 1$. Lemma~\ref{Euclidean dual BCH} gives
\[
\{0,1,\dots,2q^{\frac{m}{2}} - 2\} \subseteq \mathbf{T}_{g^\perp},
\]
and therefore
\[
d(\mathcal{C}^{\perp}_{(q,n,\delta,1,\beta)}) \ge 2q^{\frac{m}{2}}.
\] Consequently, we have
\[
d(\mathcal{C}') = \min\{2d(\mathcal{C}_{(q,n,\delta,1,\beta)}),\, d(\mathcal{C}^{\perp}_{(q,n,\delta,1,\beta)})\}
= 2(q^{\frac{m}{2}} - 1).
\]
This completes the proof.
\end{proof}

Next, we determine the parameters of Hermitian self-dual cyclic codes for odd $m$.
\begin{theorem}
    \label{Exact Hermitian}
    The Hermitian self-dual cyclic code $\mathcal{C}'$ in Lemma \ref{odd Hermitian self-dual} has parameters \[[2(q^m-1),q^m-1,2(q_1^m-1)]_q.\]
\end{theorem}

\begin{proof}
It is worthy noticing that
\[
\gcd(n, b-1) = n = q^m - 1 = (q_1^m - 1)(q_1^m + 1)
\]
and $\delta = q_1^m - 1$, we have $\delta \mid \gcd(n, b-1)$. By Lemma~\ref{weight}, this implies
\[
d(\mathcal{C}_{(q,n,\delta,1,\beta)}) = q_1^m - 1.
\]

By setting $t = \frac{m+1}{2}$ and $b = 1$, we have $\delta = q_1^{2t-1} - b = q_1^m - 1$. Lemma~\ref{Hermitian dual BCH} gives
\[
\{0,1,\dots,2q_1^m - 2\} \subseteq \mathbf{T}_{g^{\perp_H}},
\]
and therefore
\[
d(\mathcal{C}^{\perp_H}_{(q,n,\delta,1,\beta)}) \ge 2q_1^m.
\]
Consequently, we have
\[
d(\mathcal{C}') = \min\{2d(\mathcal{C}_{(q,n,\delta,1,\beta)}),\, d(\mathcal{C}^{\perp_H}_{(q,n,\delta,1,\beta)})\}
= 2(q_1^m - 1).
\]
This completes the proof.
\end{proof}

The following remarks concerning the Euclidean and Hermitian self-dual cyclic codes $\mathcal{C}'$ in Theorems~\ref{exact Euclidean} and~\ref{Exact Hermitian} are in order.

\begin{itemize}
    \item Since the inequalities 
    \[
    4(q^{\frac{m}{2}} - 1)^2 > 2(q^{\frac{m}{2}} + 1)(q^{\frac{m}{2}} - 1) = 2(q^m - 1)
    \]
    and
    \[
    4(q_1^m - 1)^2 > 2(q_1^m + 1)(q_1^m - 1) = 2(q^m - 1)
    \]
    holds, the lower bounds on the minimum distance of the self-dual cyclic codes $\mathcal{C}'$ exceeds square-root lower bounds, which replies the Open Problem in~\cite{10856234} by setting $q=2$.

    \item Theorem~\ref{Exact Hermitian} clarifies why the lower bound on the minimum distance of the Hermitian self-dual cyclic codes in Lemma~\ref{odd Hermitian self-dual} is indeed a square-root lower bound.
\end{itemize}

Table~\ref{even} lists the parameters of several Euclidean self-dual cyclic codes obtained in Theorem~\ref{exact Euclidean}.

\begin{table}[htbp]
  \centering
  \caption{Parameters of Euclidean self-dual cyclic codes for even $m$}
  \label{even}
  \begin{tabular}{|c|c|c|c|c|c|c|}
    \hline
    $q$ & $m$ & $n$ & $\delta$ & $d(\mathcal{C}_{(q,n,\delta,1,\beta)})$ & $d(\mathcal{C}_{(q,n,\delta,1,\beta)}^\perp)$ & $d(\mathcal{C}')$ \\
    \hline
    2 & 4 & 15 & 3 & $3$ & $8$ & $6$ \\
    \hline
    2 & 6 & 63 & 7 & $7$ & $16$ & $14$ \\
    \hline
    4 & 4 & 255 & 15 & $15$ & $\ge32$ & $30$ \\
    \hline
  \end{tabular}
\end{table}

\section{Conclusion}
\label{5}
This paper refines the parameter estimates for Euclidean self-dual cyclic codes with odd $\mathrm{ord}_n(q)$ and Hermitian self-dual cyclic codes with even $\mathrm{ord}_n(q)$. We notice that the inequality $2d(\mathcal{C}_1) < d(\mathcal{C}_1^\perp)$ does not always hold. Moreover, for a BCH code $\mathcal{C}_1$, lowering the designed distance $\delta$ may lead to an increase in the minimum distance of its dual $\mathcal{C}_1^\perp$. This observation suggests that the minimum distance $d(\mathcal{C})$ of the resulting self-dual cyclic code can potentially be improved by reducing $\delta$. Accordingly, we investigate the lower bounds on the minimum distances of self-dual cyclic codes obtained from new designed distance.

We establish that the self-dual cyclic codes constructed under the condition of even $\mathrm{ord}_n(q)$ exceed square-root lower bounds on their minimum distances. By leveraging known results on the defining set of the dual code and invoking the BCH bound, we derive a lower bound on $d(\mathcal{C}_1^\perp)$, which in turn yields an improved lower bound on $d(\mathcal{C})$. Following this strategy, we obtain the exact parameters of Euclidean self-dual cyclic codes with even $\mathrm{ord}_n(q)$ and Hermitian self-dual cyclic codes with odd $\mathrm{ord}_n(q)$. 

In practice, we observe a significant discrepancy between the lower bound on the minimum distance derived from the defining-set analysis in this paper and the actual minimum distance. This is because only a small subset of the defining-set zeros has been utilized. Whether a tighter bound than the Roos bound that can exploit more zeros exists remains an open problem.

This paper has primarily focused on the minimum distances of Euclidean and Hermitian dual codes. Sufficient conditions for the inequality $d(\mathcal{C}^\perp) \ge 2d(\mathcal{C})$ would be highly beneficial for estimating lower bounds on the minimum distances of related self-dual cyclic codes.

In fact, further reducing the designed distance yields additional consecutive zero segments in the defining set of the dual code. Applying the Roos bound to these extended zerso segments then gives tighter lower bounds than those obtained from the BCH bound, which is an issue that merits further investigation.

\section*{Acknowledgments}
This work is supported by the National Natural Science Foundation of China (No. 12441107), Guangdong Basic and Applied Basic Research Foundation (No.~\seqsplit{2025A1515011764}),  and  the National Key Research and Development Program of
China (No.~\seqsplit{2025YFA1017100}).


\bibliographystyle{IEEEtran}
\bibliography{reference}

\begin{IEEEbiographynophoto}{BoFeng Huang}
received the BS degree from Sun Yat-sen University, Guangdong, P.R. China. He is currently studying in mathematics at Sun Yat-sen University. His research interests lie in algebraic coding theory.
\end{IEEEbiographynophoto}

\begin{IEEEbiographynophoto}{Jingwei Zhang}
 received the B.S. degree from the Department of Mathematics, Hunan University of Science and Technology, P.R.China, in 2002, the M.S. degree from the Department of Mathematics in 2005, and the Ph.D. degree from the Department of Electrical Engineering, Sun Yat-sen University, Guangzhou, China, in 2010. She currently works with the Department of Big Data Management and Application, Guangdong University of Finance and Economics, Guangzhou, China. Her research interests include algebraic coding theory, algebraic decoding algorithms and sequences design.\end{IEEEbiographynophoto}

\begin{IEEEbiographynophoto}{Chang-An Zhao}
received the bachelor’s degree in electronical engineering, the master’s degree in applied mathematics, and the PhD degree in information science and technology all from Sun Yat-sen University, Guangzhou, P.R.China, in 2001, 2005, and 2008, respectively. He works with the School of Mathematics, Sun Yat-sen University, Guangzhou, China. His research mainly focuses on elliptic curve cryptography, post-quantum cryptography and algebraic coding theory.
\end{IEEEbiographynophoto}

\end{document}